\documentclass[11pt,twocolumn,letterpaper]{article}

\usepackage[usenames,dvipsnames]{xcolor}
\usepackage{graphicx}
\usepackage{epstopdf} 
\usepackage{amsmath}
\usepackage{amsthm}
\usepackage{nicefrac}
\usepackage{fix-cm} 
\usepackage{url}
\usepackage{macros}
\usepackage{abstract}
\usepackage{fancyhdr}
\usepackage{dblfloatfix}
\usepackage[raggedright]{titlesec}
\usepackage[margin=10pt,font=small,labelfont=bf]{caption}

\usepackage{natbib}
 \bibpunct[, ]{(}{)}{,}{a}{}{,}%
\newtheorem{theorem}{Theorem}

\newcommand{\keywords}[1]{\par\addvspace\baselineskip\noindent
                          {\bfseries Keywords:}\enspace\ignorespaces#1}

\begin{document}

\date{}

\pagestyle{fancy}
\addtolength{\headheight}{6pt}

\title{Communication Network Design:\\
       Balancing Modularity and Mixing via \\
       Optimal Graph Spectra}

\rhead{\bfseries Communication Network Design}

\author{
\large
\textsc{Benjamin Lubin,$^{a}$ Jesse Shore$^{a}$, and Vatche Ishakian$^{b}$}\\[2mm]
\normalsize 
$^{a}$Boston University School of Management \hspace{.5em} 
$^{b}$BBN Technologies\\
}

\lhead{\itshape Lubin, Shore, and Ishakian}   

\twocolumn[
\vspace{-1cm}
\maketitle

\vspace{-1cm}

\begin{onecolabstract}
%

\noindent By leveraging information technologies, organizations now
have the ability to design their communication networks and
crowdsourcing platforms to pursue various performance goals, but
existing research on network design does not account for the specific
features of social networks, such as the notion of teams.
We fill this gap by demonstrating how desirable aspects of
organizational structure can be mapped parsimoniously onto the
spectrum of the graph Laplacian allowing the specification of
structural objectives and build on recent advances in non-convex
programming to optimize them.
This design framework is general, but we focus here on the problem of
creating graphs that balance high modularity and low mixing time, and
show how ``liaisons'' rather than brokers maximize this objective.

\end{onecolabstract}


\vspace{1cm}
]


\section{Introduction}
\label{sec:intro}

As organizations move rapidly into the world of crowdsourcing,
external innovation and human computation \citep{aral2013introduction,
  boudreau2010open,boudreau2011incentives,di2010getting,Guinan2013exp,
  ipeirotis2010quality, kearns2012experiments,kohler2011co,
  lakhani2007r,lakhani2007principles,von2009human,zeng2013social,
  zheng2011task}, they are creating platforms that govern the
relationships among their problem solvers.  When individuals are
connected on such platforms, they gain the ability to observe,
communicate, and collaborate with each other, forming communication
networks.  The structure of these networks is very important, being
well known to influence knowledge management \citep{alavi2001review},
knowledge sharing
\citep{aral2011diversity,borgatti2003relational,mcevily2005embedded},
cooperation and coordination
\citep{huang2011critical,mccubbins2009connected,suri2011cooperation},
innovation
\citep{bae2011cross,capaldo2007network,reagans2001networks}, the
balance between exploration and exploitation
\citep{lazer2007network,mason2008propagation, mason2012collaborative,
  shore2013oswc}, and overall problem-solving performance
\citep{bavelas1950communication,cummings2003structural,sparrowe2001social}.
By leveraging modern information and communications technology, there
is now the opportunity for organizations to go beyond
\emph{understanding} their networks to actually \emph{designing} their
networks.

There are undoubtedly roles for both emergent and designed networks,
but there are reasons to doubt whether emergent networks, created as
individuals pursue their own goals, are optimal in the aggregate for
the whole organization or collective.  For example, if all actors seek
novel information by forming bridging ties, this could quickly drive
out diversity of information at the collective level as macroscopic
structural divisions in a network disappear \citep{gulati2012rise}.
Additionally, in resource networks, if individuals keep only their
strongest ties, it can paradoxically make the whole system very weak
\citep{shore2013power}.  In other words, networks can be subject to
social dilemmas in which individual and collective interests are at
odds.  In these cases, design of network structure is an attractive
option for the organizations that control them via digital platforms.

Unfortunately, however, there is little research specifically
addressing the design of networks of communicating human beings that
might guide the way. Rather, the design literature has focused on
problems of minimal cost or optimally ``efficient'' networks, with
applications in non-human infrastructure settings
\citep{balakrishnan1989dual,dionne1979exact,donetti2005entangled,
  estrada2007graphs,guimera2002optimal,kershenbaum1991mentor,
  lubotzky1988ramanujan,magnanti1984network,minoux1989networks,
  winter1987steiner}.  The work of \cite{lovejoy2010efficient} is a
notable exception in that it is concerned with social networks within
organizations, but it is similar in its orientation toward efficiency
and short paths between any given pair of individuals in a network.
There is indeed substantial theoretical justification for targeting
short paths as a design criterion in human as well as infrastructural
networks that is generally understood in terms of two related ideas:
that weak-ties enable rapid diffusion of information
\citep{watts1998small} and that bridging structural holes can be
associated with innovation \citep{burt2004structural}.

Although these are important issues, there are also advantages to
modularity -- having teams or groups in organizations that are
relatively separate but internally cohesive -- but this has to our
knowledge been omitted as a network design criterion.  Within
organizations, internally cohesive groups tend to use similar language
constructs, which enables high-bandwidth communication
\citep{aral2011diversity} and increases their effectiveness
\citep{hansen1999search,reagans2001networks}, especially for problems
that require extensive information-space searching or coordination
\citep{shore2013oswc}. Additionally, certain types of information and
behaviors spread more easily within rather than between clusters of
connected individuals \citep{centola2010spread}. Finally, real
organizations are usually structured in divisions, work groups, or
teams --- lending an added importance to incorporating some notion of
modularity into network design work.  Despite all of this, the design
literature has yet to address network contexts in which modularity is
desirable, keeping design off of the table for most applications to
human organizations.

Two major issues may have stood in the way of incorporating modularity
into design work.  First, obtaining modularity and short path lengths
imply quite different network structures, making theoretical analysis
that encompasses both properties difficult. Second, the space of all
possible networks is combinatorially large, making the design problem
formidably complex (for example, the number of possible undirected
graphs with 16 nodes is $2^{120}$, or approximately $1.3 \times
10^{36}$ --- far too many to evaluate individually by any known
means).

Here, we propose a design framework that addresses both issues
simultaneously: we frame the network design problem in a way that lets
the designer trade off between modularity and mixing time, and we
propose an algorithm that can find optimal or near-optimal graphs
under these criteria. Specifically, for the design framework, we take
advantage of prior literature in the area of spectral graph theory and
demonstrate how desirable aspects of organizational structure can be
mapped parsimoniously onto the spectrum of the graph Laplacian
\citep{chung1997spectral} derived from a matrix representation of that
communication structure.  Recent advances in convex and non-convex
optimization allow us to capture these spectral elements in an
objective function to be optimized.  We go on to present examples of
the communications structures produced under this method that balance
modularity and mixing time and discuss the implications of their
properties.  Specifically, rather than structural ``brokerage,'' we
find that networks with a ``liaison'' structure offer more modularity
for a given mixing velocity.  Beyond our specific results, however,
our more general contribution is to provide a framework for both
specifying and solving principled problems in network design.


\section{Spectral Theory Informs Design}
\label{sec:theory}

Spectral graph theory \citep{chung1997spectral,cvetkovic1980spectra}
is concerned with the relationships between the structure of a network
and the eigenvalues, also called the ``spectrum,'' of the matrix
representation of that network.
One major advantage of thinking of networks in terms of their spectra
is that spectra are insensitive to permutations and labeling.  All
networks with the same structure have the same spectrum. This property
lets us avoid having to deal with the so-called ``graph isomorphism
problem,'' where many equivalent representations for structurally
isomorphic graphs exist, making search and classification in graph
space difficult.
In essence, working with spectra lets us focus on more tractable and
compact objects, which correspond to unique graphs with high
probability (see section~\ref{sec:cospectral}).
Moreover, the values of the spectrum provide enormously useful
information about graph structure in a compact and accessible way.
These properties make spectra ideal mathematical objects to use in
formalizing desiderata and constraints in network design problems.

In this paper, we adopt a particular design objective: we aim to
design networks that both manifest distinct subgroups and yet are
still ``sufficiently connected.''  As we have seen in the previous
section, these are well motivated goals.  However, it is not obvious
how to formalize them.  Spectral theory gives us a means to frame this
precisely.  Existing work has not examined such an objective; we
provide:

\begin{itemize}
\item A spectral formalization of our modularity and mixing objective
  (section~\ref{sec:laplacianspectrum})
\item A novel optimization problem based on this formulation that
  captures our design objective (section~\ref{sec:formulation})
\item An algorithm for approximately (and often optimally) solving
  this problem (section~\ref{sec:algorithm})
\item A set of numerical experiments based on this algorithm, and
  their results and interpretation (sections~\ref{sec:results}
  and~\ref{sec:discussion}).
\end{itemize}

\subsection{Preliminaries}

The standard matrix representation of a graph, where each entry
represents the strength of the connection between the node indexed by
the matrix row and column, is called the \textit{adjacency} matrix.
In this paper, we assume that each individual in the organization has
equal capacity to communicate that they use fully.  This implies that
our matrix representations of the network must have rows and columns
that can be normalized so that they all sum to 1; such matrices are
called ``doubly stochastic.''  Further, we assume that a given
communication tie takes the same proportion of each connected
individual's communication capacity.\footnote{Formally, this property
  holds that the rows are scalar multiples of the columns.}  Together,
these properties imply that the matrix representation of the network
must be symmetric about its diagonal.

\begin{figure*}[htb!]
  \centering
    \includegraphics[width=0.60\textwidth]{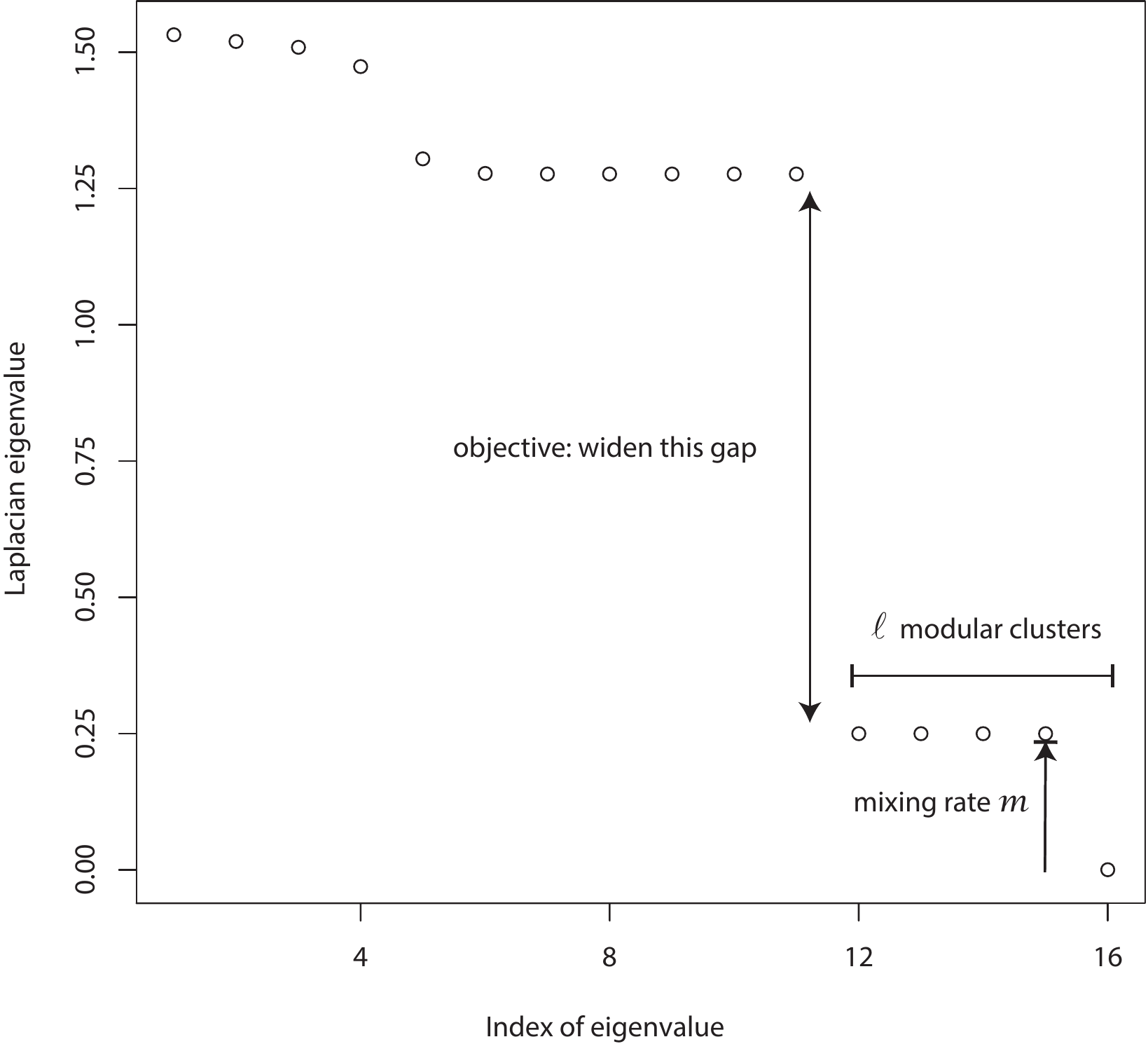}
  \caption{Illustration of the spectral framework, including objective
    and constraints}
  \label{fig:framework}
\end{figure*}

Instead of working with the adjacency matrix, it can be useful to work
with the graph \textit{Laplacian} matrix given, for stochastic graphs,
by $L = I - A$, where $I$ is the identity matrix and $A$ the adjacency
matrix.\footnote{In general the \textit{Laplacian} is given by $L = D
  - A$, where $D$ is the degree matrix, constructed by putting the row
  sums of $A$ on the diagonal, with zeros elsewhere.}
The spectrum of the adjacency and Laplacian matrices are related but
have distinct properties; those of the Laplacian match our needs, and
consequently we adopt it here.  The matrix spectrum is simply the
multiset of eigenvalues, sorted in decreasing order of magnitude.%
\footnote{ The eigenvalues of a matrix $M$ are given by $\{\lambda |
  M\bv{v}=\lambda \bv{v}, \bv{v} \neq 0\}$.  The $\bv{v}$ are called
  the eigenvectors of the matrix: those vectors that when multiplied
  by the matrix yield a scaled copy of themselves.  Each scale factor
  is a corresponding eigenvalue, $\lambda$.}
Such a spectrum can be plotted as a set of points, as illustrated in
Figure~\ref{fig:framework} and elaborated upon below.

\subsection{The Laplacian Spectrum and Network Structure}
\label{sec:laplacianspectrum}

The relative magnitude of the various spectral values correspond to
specific structural properties of the corresponding network.  We
describe those necessary for capturing our design objective below.

\subsubsection{Bounding Mixing Time with $\boldsymbol{m}$}


The magnitude of the smallest Laplacian eigenvalue (hereafter, just
``eigenvalue'' for brevity) is always zero, and therefore of little
immediate interest.  However, the magnitude of the second smallest
eigenvalue is also the graph's ``algebraic connectivity''
\citep{fiedler1973algebraic} and is inversely related to the mixing
time for Markov chains \citep{mohar1997some}.  In short, the larger
the second smallest eigenvalue, the faster we expect information to
diffuse through the network \citep{donetti2006optimal}.  Because of
its known connection to mixing time, we refer to the magnitude of the
second smallest eigenvalue as $m$ (see Figure~\ref{fig:framework}).
By tuning $m$, a network designer has a spectral method for
formalizing the idea of ``sufficiently connected:'' the larger the
$m$, the more rapidly that communication structure is expected to
diffuse information.  However, raising $m$ may come at the cost of
other desirable features, such as the amount of modularity that is
manifest in the network, as we shall see shortly.

\subsubsection{Setting the Number of Modular Clusters with 
$\boldsymbol{\ell}$}

It is well known that the number of connected components of an
undirected graph is equal to the number eigenvalues of the Laplacian
that are equal to zero \citep{brouwer2011spectra}.  For example, if
there were four totally disconnected components, there would be four
eigenvalues equal to zero.  If, however, there existed weak
connections among those distinct communities such that they are no
longer disconnected components but rather modular clusters, then
rather than having one zero for each cluster, we would have one small
eigenvalue for each module \citep{donetti2006optimal}.  Consequently,
for a graph consisting of four modular clusters that are weakly
connected to each other, the spectrum of the Laplacian (hereafter
``spectrum'') would contain four small eigenvalues, one of which would
be zero (as there would be one component, and thus one eigenvalue
equal to zero).

From the design point of view, then, we observe that if one desires a
communication network with some number, $ \ell $, distinct modular
clusters, then one should construct a graph with a spectrum containing
$ \ell $ small eigenvalues, one of which is zero (see Figure
\ref{fig:framework}).


\subsubsection{Maximizing Modularity with the Rest of the Spectrum}

We have just argued that to generate $\ell$ modular clusters we want
$\ell$ small eigenvalues because a structurally separate portion of
the graph will be established for each small eigenvalue present in the
spectrum.  By the same logic, if we want the next most significant
structural division to be minimized, we should seek to make the
$\ell+1^{th}$ eigenvalue as large as possible.

We then consider the effect of raising the $\ell+1^{th}$ eigenvalue on
the balance of the spectrum.  We note that the sum of all of the
eigenvalues is constrained to equal to $n$
\citep[][p. 6]{chung1997spectral}\footnote{As long as there are no
  isolated vertices, which will be the case for us} and consequently:
\begin{equation}
\label{eq:spectralSums}
\sum_{\ell < k \leq n} \lambda_k = n - \sum_{1 \leq k \leq \ell} \lambda_k
\end{equation}
Further, suppose we hold the $\ell$ smallest eigenvalues constant, as
will happen when they are against their $m$ lower bound.  Then the
right side of equation~\ref{eq:spectralSums} will be a known constant.
In general, driving up the smallest of a set of $n$ numbers whose sum
is constrained at some number $c$, will push these numbers to each
take the value $\nicefrac{c}{n}$ (see Theorem~\ref{th:optBound} below
for more detail).
Such constant spectra correspond to homogeneously connected
networks \citep[][p. 6]{chung1997spectral}.  Thus, because our
maximization will tend to equalize the large eigenvalues, we are
minimizing any additional remaining structural divisions.

A theorem provided by Newman and Kel'mans enables us to add an
additional intuition as to why this is true
\citep{newman200laplacian}:
\begin{equation}
\lambda_k(G^C) = n - \lambda_{n+2-k}(G) \textrm{ for } 2 \leq k \leq n
\end{equation}
where $G$ is a graph and $G^C$ is its complement.\footnote{Informally,
  the complement has weight wherever the primary graph does not and
  vice versa.}  This theorem provides that the $k^{th}$ largest
eigenvalue is equivalent to the $k-1^{th}$ smallest eigenvalue of the
complementary graph.
So, by driving down the largest eigenvalue in the primary graph,
$\lambda_n$, we are driving up the smallest non-zero eigenvalue,
$\lambda_2$, of the complementary graph, increasing its mixing rate.
This then, in the primary graph is equivalent to increasing modularity
by reducing the between-module weight as much as possible.

Thus in sum, the larger the $\ell+1^{th}$ eigenvalue, the more modular
the resulting graph, holding the first $\ell$ eigenvalues constant.

\subsection{Co-Spectral Graphs}
\label{sec:cospectral}
It is one thing to calculate the spectrum of a known graph and quite
another to construct a graph with a given spectrum.  We are more
concerned with the latter problem.  The next section details our
method for constructing matrices with desirable spectral properties.
Before we do so, however, we must take note of the issue of
co-spectral graphs, or non-isomorphic graphs with the same spectrum
\citep{godsil1982constructing,harary1971cospectral}.

Although at present relatively little is known about which graphs have
co-spectral partners \citep{van2003graphs}, we do not believe this
presents an impediment to the present undertaking.  Most
fundamentally, we are presenting a framework for designing
communication networks with properties that have spectral correlates.
If by chance we construct a graph for which there exists a co-spectral
partner that we do not find, we will have still achieved our design
goal, because co-spectral graphs have similar structure with respect
to the features captured by that spectrum.

Additionally, but less essentially, enumerations of unweighted graphs
that are co-spectral with respect to their Laplacian
\citep{brouwer2009cospectral,cvetkovic2012spectral,haemers2004enumeration}
show that the proportion of graphs with co-spectral partners is
highest at $n=9$ and decreases as $n$ and the number of edges
increase. \cite{halbeisen2000reconstruction} show that for weighted
graphs --- which we employ here --- there are almost surely no
co-spectral partners.  Therefore, we assert that by constructing
weighted networks according to spectral parameters, we are not leaving
anything important to our aims on the table.


\section{Methods}

\label{sec:formulation}

Spectral theory has given us the means to formalize both of our design
objectives:
\begin{itemize}
\item Sufficient connectivity, by imposing a lower bound, $m$, on the
  second smallest eigenvalue $\lambda_2$, which ensures a fast enough
  mixing time.
\item Modularity with $\ell$ clusters, by having $\ell$ small
  eigenvalues and $n-\ell$ large eigenvalues.
\end{itemize}
Our network design problem can then be cast as the following
non-linear optimization problem:
\begin{align}
\max_{\bm{W}} \quad \lambda_{\ell+1}(\bm{W})& -\lambda_{\ell}(\bm{W}) 
                                                          \label{eq:obj}\\
\textrm{s.t.} \quad\quad\quad\quad\; \lambda_2 & \geq m   \label{eq:mConst}\\
              \sum_j \bm{W}_{ij} & = 1 \; \forall \; i \label{eq:stochastic} \\
              \bm{W}_{ij} & = \bm{W}_{ji} \; \forall \; i,j \label{eq:symmetric}
\end{align}
\noindent where $\lambda_k(\bm{W})$ is the $k^{th}$ eigenvalue of
the matrix $\bm{W}$.
The objective, equation~\ref{eq:obj}, maximizes the difference between
the $\ell+1^{th}$ and $\ell^{th}$ Laplacian eigenvalue.
Constraint~\ref{eq:mConst} ensures that the mixing rate is at least
$m$.  Constraints~\ref{eq:stochastic} and~\ref{eq:symmetric} ensure
stochasticity and symmetry respectively.  Note that the variables in
this formulation are the weights of matrix $\bm{W}$.

\subsection{ Optimization Algorithm }
\label{sec:algorithm}

The ``eigenvalue problem,'' that of computing the eigenvalues of a
known matrix, can be calculated in closed form for small matrices, and
for large matrices by numerical algorithms, for example QR, that have
been known since the early sixties
\citep{francis1961qr,francis1962qr}.  However, the ``inverse
eigenvalue problem,'' that of finding the graph that corresponds to a
specific spectrum or to specific spectral characteristics has proven
vastly harder to solve \citep{chu1998inverse}.  Most such problems
admit no computationally tractable algorithm for obtaining a globally
optimal solution.

Our formulation falls within this hard class, and thus the best we can
hope for is a high-quality approximation algorithm.  We are not aware
of any existing work that has looked at solving our particular
spectral objective and constraints.  We have therefore constructed our
own solution method by leveraging recent advances in Semi-Definite
Programming (SDP) and Difference in Convex (DC) programming, which we
describe next.

\subsubsection{Semi-Definite Programming}

SDP is a type of convex optimization that operates over a matrix
variable, instead of the scalar variables seen in other convex
optimization methods \citep{vandenberghe1996semidefinite}.  SDP
objectives are specified as the inner product of the matrix variable,
with a user-specified constant matrix.  Similarly, SDP constraints
consist of a bound on the inner product between the matrix variable
and another user-specified constant matrix.  The minimal value for the
objective is found, where the matrix variable is drawn from the cone
of semi-definite matrices.  Many problems can be cast into this
structure, and because the resulting formulation is convex, it can be
solved efficiently by, for example, interior point methods
\citep{alizadeh1995interior,todd2001semidefinite,wolkowicz2000handbook}.

For the present work, the key property of SDP is its ability to
capture the sum of the $k$ smallest Laplacian eigenvalues,
$S_k=\sum_{1\leq i \leq k} \lambda_i$, as a concave function of the
matrix weights, the maximization of which is a convex optimization.
\cite{boyd2006convex} and colleagues use this capability to solve
certain Laplacian inverse eigenvalue problems directly
\citep{boyd2004fastest,boyd2006convex}.  For example, they formulate
$S_2$ as a concave function which they maximize via SDP to efficiently
solve for the matrix corresponding to the Markov process with the
fastest mixing time.  We leverage their result by moving their
objective formulation to a constraint, obtaining a convex form for
equation~\ref{eq:mConst}.  Further, as the remaining constraints are
linear, only our objective~\ref{eq:obj} fails to be directly
representable as an SDP, which we address next.

We start by noting that $\lambda_\ell = S_\ell - S_{\ell-1}$ and
$\lambda_{\ell+1} = S_{\ell+1} - S_{\ell}$.  And thus our objective in
equation~\ref{eq:obj} can be rewritten as:
\begin{equation}
\begin{split}
\lambda_{\ell+1} - \lambda_\ell 
 &= (S_{\ell+1} - S_\ell) - (S_\ell - S_{\ell-1}) \\
 &= S_{\ell+1} + S_{\ell-1} - 2S_\ell \label{eq:summationForm}
\end{split}
\end{equation}
This objective captures our intent, and can be formulated by known
SDP-style expressions.  However, it can not be directly solved
because, when treated as a maximization and not a minimization, the
third term is non-convex.

\subsubsection{Difference in Convex Programming}

As we have seen, the formulation of equation~\ref{eq:obj} given in
equation~\ref{eq:summationForm}, is almost convex and solvable as an
SDP, but not quite.
Consequently, we are not going to be able to directly use convex
optimization, and the best we can hope for is an approximately optimal
algorithm.  However, equation~\ref{eq:summationForm} is a
\textit{difference of convex functions} and, as such, is amenable to
an algorithm known as the Concave-Convex Procedure
\citep{yuille2002concave,yuille2003concave}.  This is an iterative
method for obtaining approximate solutions to problems with convex and
concave components in the objective that has good convergence
properties \citep{sriperumbudur2009convergence}.  Our approach is to
implement the Concave-Convex Procedure over our SDP
formulation.\footnote{The Concave-Convex Procedure has generally been
  used for simpler optimization formalisms in the literature, here we
  adapt it to the more expressive SDP context.}  Our approach is as
follows:

We start with a random initial graph $\widehat{\bm{W}}$.  We then form a
first-order Taylor expansion of the concave portion of the objective
around $\widehat{\bm{W}}$.  Using this linear form, we can then
approximate the objective as:
\begin{multline}
  S_{\ell-1}(\bm{W}) + S_{\ell+1}(\bm{W}) - \\ 
  2 \left(  
  S_\ell(\widehat{\bm{W}}) + 
  \nabla S_\ell(\widehat{\bm{W}}) \cdot (\bm{W} - \widehat{\bm{W}}) 
  \right)
\end{multline}
This then, is directly representable as an SDP, which we solve using
the CVX package \citep{cvx,gb08}.  We then set $\widehat{\bm{W}}
\leftarrow \bm{W}$ and repeat until convergence.

\subsection{Bounding the Objective Value}
\label{sec:optBound}

\begin{figure*}[tbh!]
	\centering
        \includegraphics[width=\textwidth]{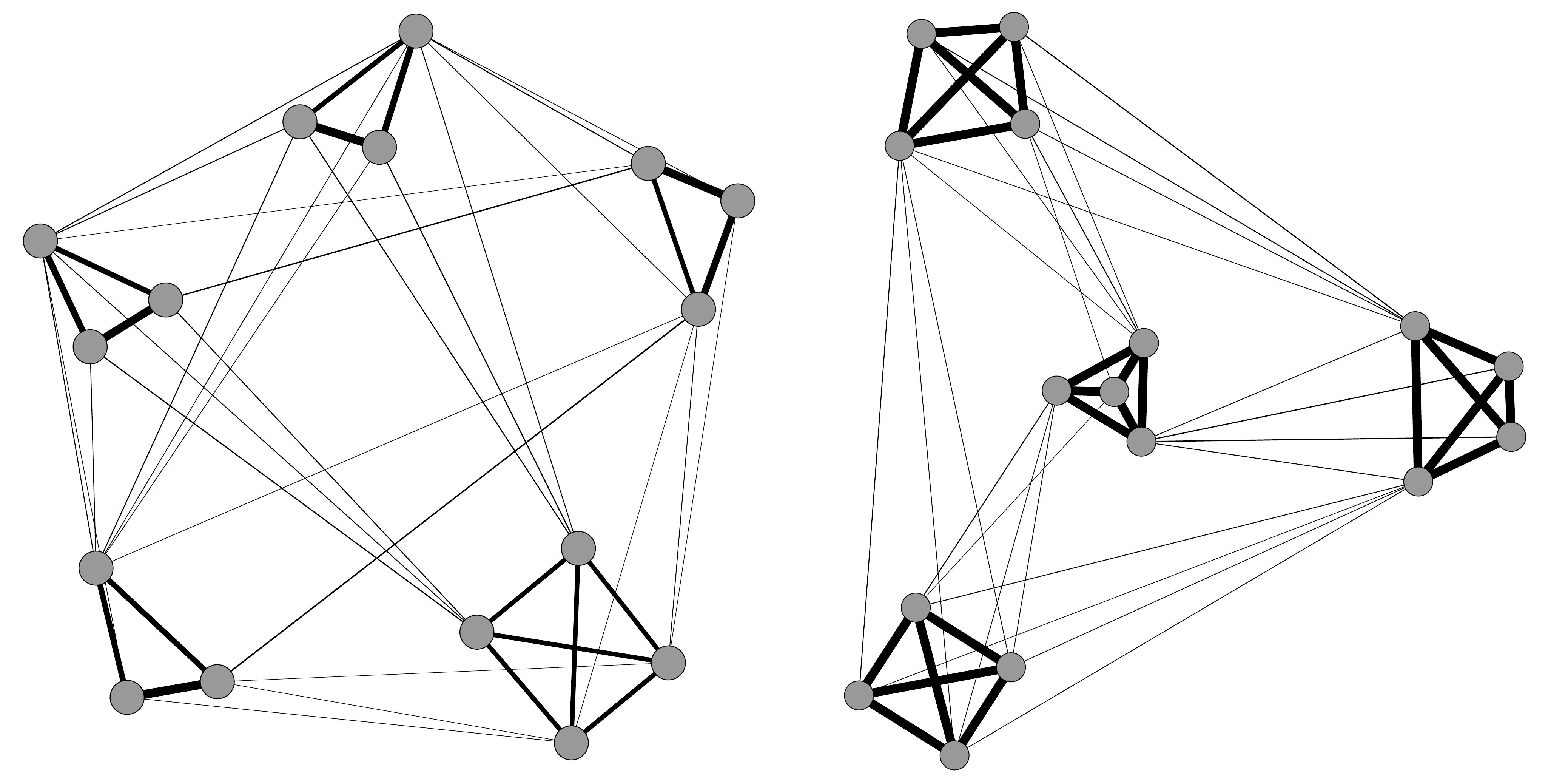}
        \caption{Two examples of sixteen-person
          communication networks produced by the spectral design
          framework.  The left hand network has \( \ell = 5 \) and the
          right hand \( \ell = 4 \) teams, both with mixing rate \( m
          = 0.25\).}
	\label{fig:clearExamples}
\end{figure*}

Because the solution found by our optimization algorithm may be only
locally optimal, it is useful to have a theoretical upper bound on the
objective value in equation~\ref{eq:obj}.  When the objective value of
the solution found by our numerical calculations approaches the bound,
we have found an approximately optimal graph.\footnote{However, the
  converse does not necessarily follow: solutions far from the bound
  may still be near-optimal when the bound is loose.} Accordingly, we
can take advantage of the following:
\begin{theorem}
\label{th:optBound}
  $\frac{n-m(\ell-1)}{n-\ell}-m$ gives an upper bound on the
  non-convex objective in equation~\ref{eq:obj}.
\end{theorem}
\begin{proof}$\lambda_1=0$ always and $\lambda_k \geq m$ for
$2 \leq k \leq \ell$ by constraint~\ref{eq:mConst}, a lower bound on
each of these eigenvalues.  This implies $\sum_{1 \leq k \leq \ell}
\lambda_k \geq m(\ell-1)$.  There is a known result that $\sum_{k}
\lambda_k \leq n$ \citep{chung1997spectral}.  Subtracting the first
from the second yields $\sum_{\ell+1 \leq k \leq n} \lambda_k \leq n -
m(\ell-1)$, an upper bound on the large eigenvalues.
The smallest of these, $\lambda_{\ell+1}$, is made maximal at this
bound and when these eigenvalues are of equal size, giving it a value
of $\frac{n - m(\ell-1)}{n-\ell}$.
Subtracting our upper bound on $\lambda_{\ell+1}$ from our lower bound
on $\lambda_\ell$ gives an overall objective upper bound of:
$\frac{n-m(\ell-1)}{n-\ell}-m$.%
\end{proof}


\section{Results}
\label{sec:results}

We next describe several experiments we have conducted to find
approximately optimal graphs according to our spectral design
framework.

\subsection{Properties of Spectrally Designed Communication Networks}

Figure \ref{fig:clearExamples} shows two examples of networks produced
by our framework, with the weakest ties omitted for visual clarity.
Several features are immediately apparent. As expected, these networks
have a modular structure, with strong intra-team connections.
Additionally, there are weak ties connecting the teams in patterns
that appear in the visualization as ``fans.''  Intuitively, one could
think of these fans as ties from one representative of a team to
(usually) all the members of another team --- more of a ``liaison''
than a broker \citep[e.g.][]{burt2001structural}.  We are not aware of
graphs similar to our results appearing previously in the literature
on network structure.

The global disposition of the inter-team liaisons has a definite
structure, suggestive of a hierarchical ``spiral'' in the
visualizations.  In the right hand side of
Figure~\ref{fig:clearExamples}, the central team has three
``outgoing'' liaisons; the team to the right has two; the team at the
top has one; and the team at the bottom has none.

Interestingly, although our design criteria do not include hierarchy
or centralization, our results nonetheless show these features.  With
three out-going liaisons, the central team is at one end of the
hierarchy, perhaps at the top.  In this team, there is a leader-like
individual who is not a liaison to any other team.  However, this
individual has two singleton weak ties to individuals in other teams,
shortening the ``leader's'' paths to much of the rest of the network.

\begin{figure}[tb!]
  \centering
  \includegraphics[width=.5\textwidth]{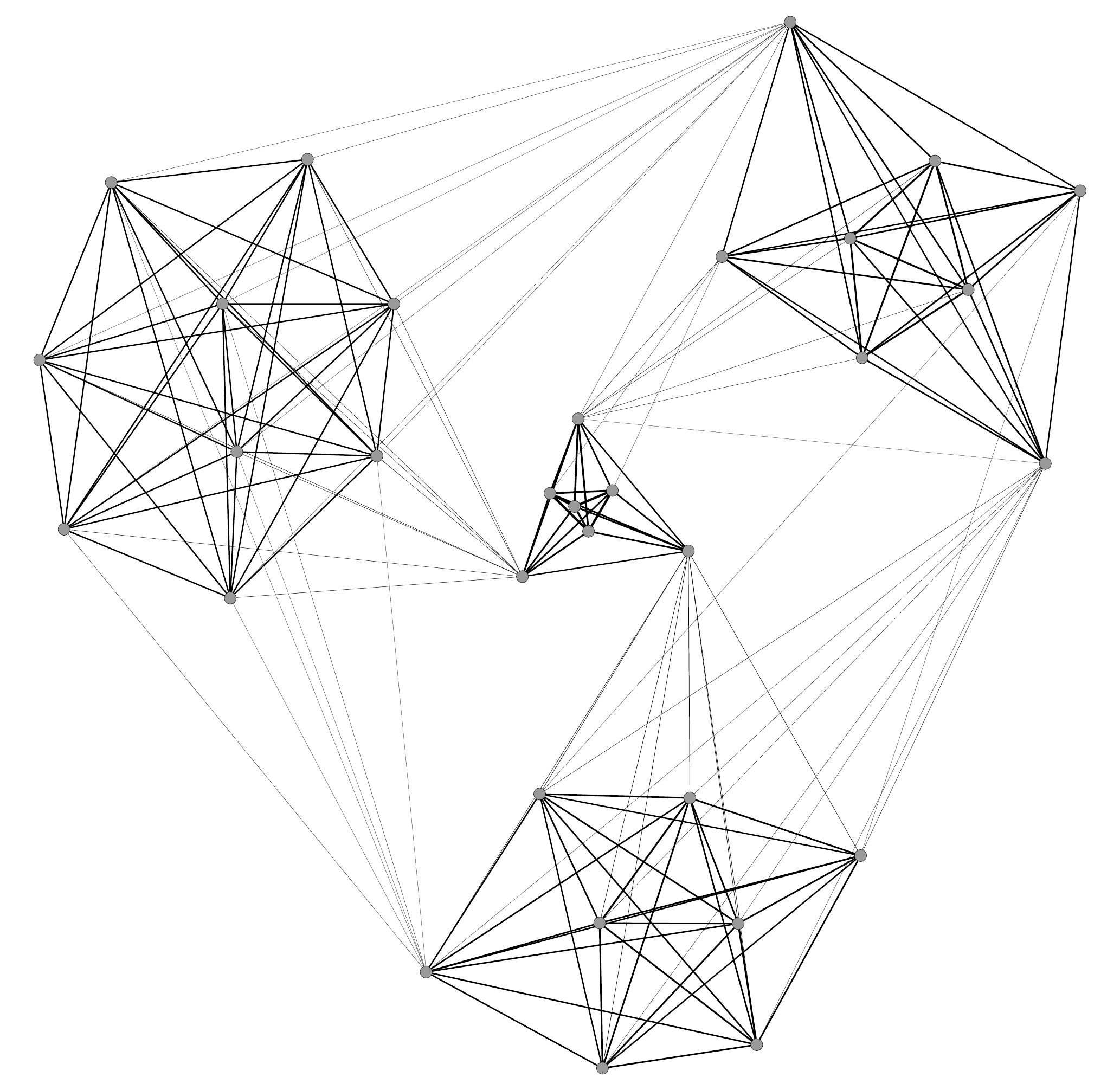}
  \caption{A 32-person network with \(\ell=4\) and \(m=0.15\) }
  \label{fig:32-28}
\end{figure}

Results for 32-person networks are similar to those for 16-person
networks. Figure \ref{fig:32-28} shows a network comparable to the
16-person network on the right-hand side of Figure
\ref{fig:clearExamples}, displaying the same hierarchical spiraling
structure.

\subsection{Optimality of Results}

\begin{table}[tbh!]
\centering
\caption{Provable achieved optimality, varying $n$, $\ell$
  and $m$.  The optimality column is calculated by dividing the
  objective value of the best solution we found by the theoretical
  upper bound described in section~\ref{sec:optBound}.  Thus, we are
  reporting a lower bound on the optimality of our solution.}
 \label{tab:optVal}
\small
\renewcommand{\arraystretch}{.9}
\begin{tabular}{rrrr}
  \hline
  $n$ & $\ell$ & $m$ & optimality\\ 
  \hline
  16 & 6 & 0.15 & 0.938 \\ 
  16 & 6 & 0.20 & 0.938 \\ 
  16 & 6 & 0.25 & 0.938 \\ 
  16 & 5 & 0.15 & 0.917 \\ 
  16 & 5 & 0.20 & 0.919 \\ 
  16 & 5 & 0.25 & 0.922 \\ 
  16 & 4 & 0.15 & 1.000 \\ 
  16 & 4 & 0.20 & 1.000 \\ 
  16 & 4 & 0.25 & 0.999 \\ 
  32 & 8 & 0.15 & 0.937 \\ 
  32 & 8 & 0.20 & 0.939 \\ 
  32 & 6 & 0.20 & 0.975 \\ 
  32 & 4 & 0.15 & 0.984 \\ 
  32 & 4 & 0.20 & 0.978 \\ 
   \hline
\end{tabular}
\end{table}

\begin{figure*}[tbh!]
  \centering
  \includegraphics[width=1.00\textwidth]{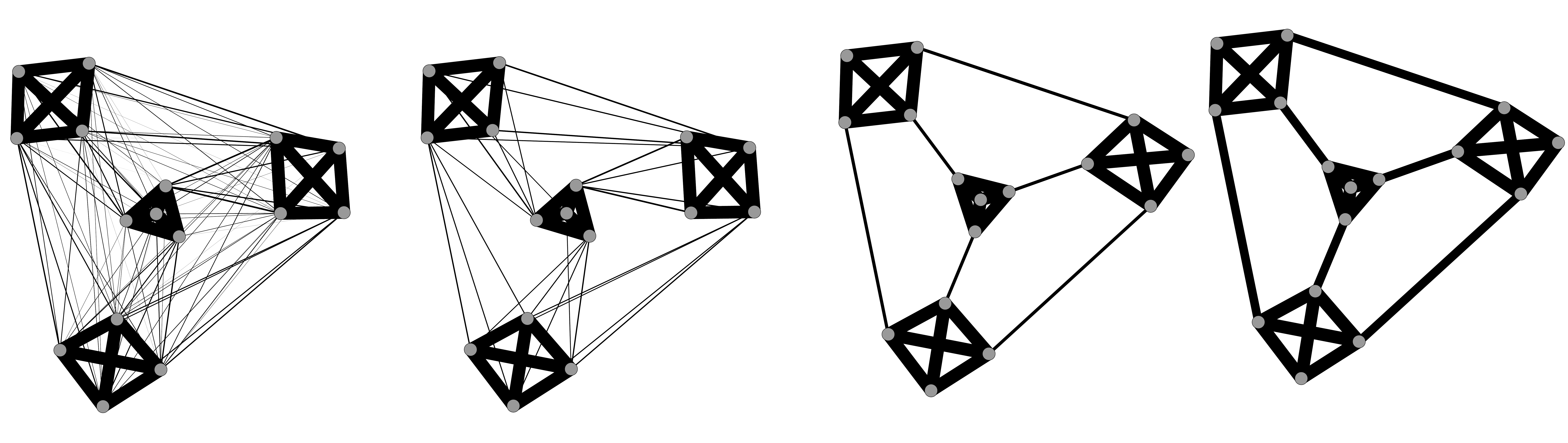}
  \caption{Four related networks, from left: 1) the full result of the
    optimization algorithm, 2) the result truncated such that the weakest
    links are omitted, 3) a hand-constructed network with brokers with
    single ties between teams and with weight equal to the sum of a
    liaison's weak ties, 4) a hand constructed network with stronger ties
    between cliques.}
  \label{fig:fourNets}
\end{figure*}

Table \ref{tab:optVal} gives the best objective values we achieved,
divided by the theoretical upper bound calculated for each set of
parameter values, as described in section~\ref{sec:optBound}.  Each
data point represents the best of 2000 independent random starting
points of our algorithm, obtained by running each in parallel on a
100 node computational cluster.  Each such pass through the algorithm
generally completes in less than 2 hours of CPU time on modern
Xeon-class hardware.
%
%
From the table we see that our algorithm is finding answers that are
very near our bound in most of the cases.  Where some gap remains, it
is unknown if this is due to the bound being loose for these
combinations of parameters, or the algorithm failing to find a
sufficiently global optimum.

\subsection{Spectral Impact of Inter-Team Connection Types}

In this section we compare our results with similar and simpler
networks.  Before doing so, however, we address the following
methodological detail: our algorithm's output is doubly stochastic,
but other networks are not necessarily so, preventing direct
comparison of their spectra.  In order to address this, we
re-normalize by iteratively row-normalizing then averaging the
resulting matrix with its transpose until we have achieved
double-stochasticity.

\subsubsection{Full versus Truncated Result}  

As noted above, the visualizations in Figure \ref{fig:clearExamples}
truncate the algorithm output such that the very weakest links are not
drawn.  In addition to making the structure of the networks more
apparent to the eye, simplified versions of the full result would
certainly be easier to implement in practice.  The two left-hand
networks in Figure \ref{fig:fourNets} visualize the difference between
the full and truncated results.

\begin{figure*}[tbh!]
  \centering
  \includegraphics[width=.60\textwidth]{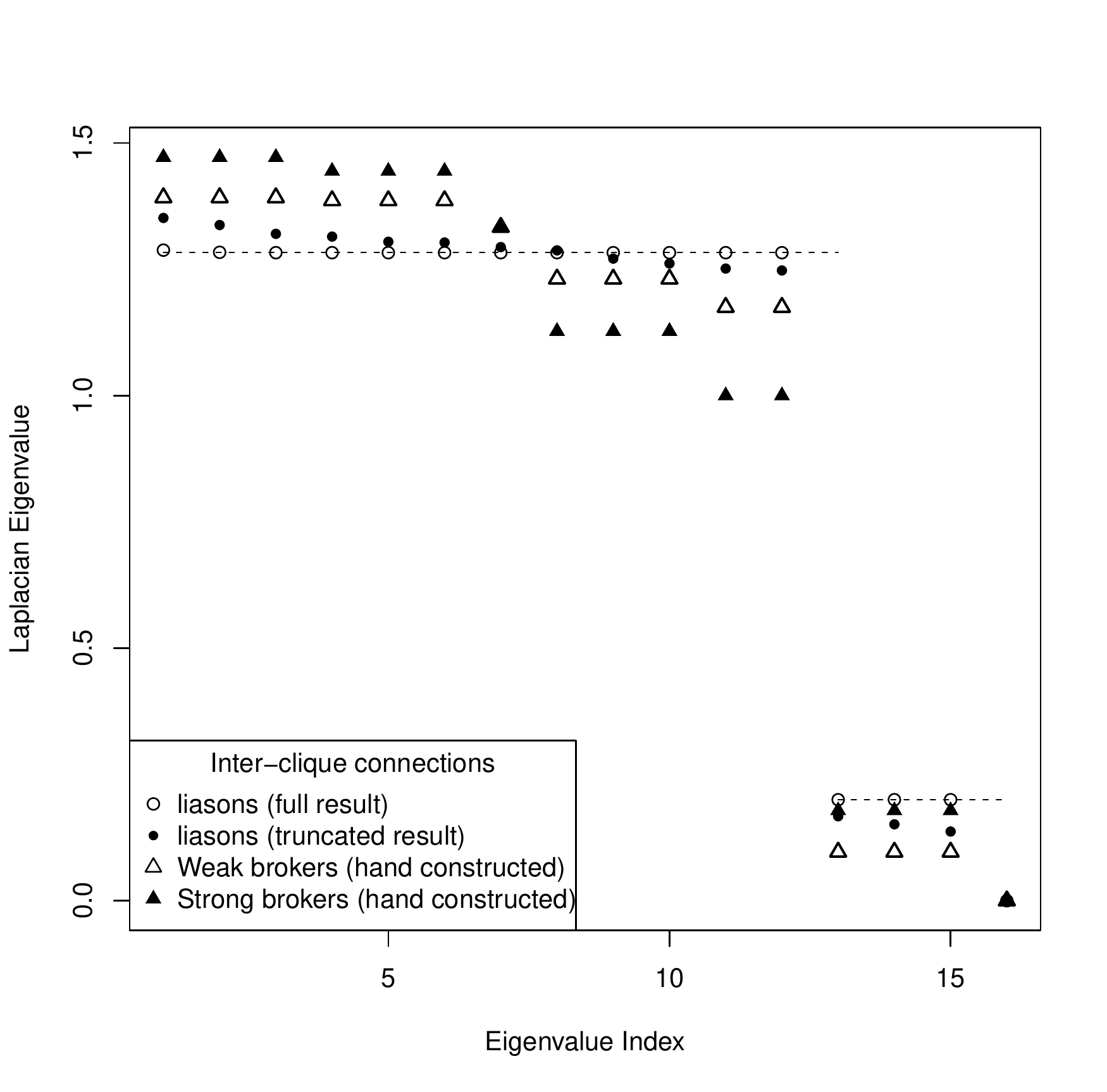}
  \caption{Comparison of spectra of related graphs in
    Figure~\ref{fig:fourNets}.  N.B. for the truncated version of the
    results and the hand-constructed matrices, weights are
    re-normalized for double-stochasticity.  The dotted line
    represents the bound described in section~\ref{sec:optBound}.}
  \label{fig:brokersVliaisons}
\end{figure*}		

Figure \ref{fig:brokersVliaisons} shows the effect of such truncation
on the spectrum.  The full solution is essentially optimal; truncation
produces a slightly sub-optimal spectrum, but the deviation is
relatively minor.  The right-hand side of the truncated spectrum shows
that the mixing rate is lower than the full result, and the left hand
side shows slight deviation from maximum modularity.  In sum, although
truncation moves us a step from the theoretical optimum, it is a small
step and a more pragmatic alternative where implementation is
important.

\subsubsection{Liaisons versus Brokers} 

Literature on social networks has tended to focus on the role of
brokers in organizations \citep{burt2001structural}, rather than the
liaisons we describe here.  In order to compare the spectra of
networks connected by brokers, rather than liaisons, we
hand-constructed two such networks.  In both cases, we replaced
inter-team liaisons with brokers.  In the first (third from the left
in Figure \ref{fig:fourNets} and the open triangles in Figure
\ref{fig:brokersVliaisons}), we set the weight of the inter-team ties
to be equal to the sum of the weights in a liaison's weak ties.  In
the second (right-hand side of Figure \ref{fig:fourNets} and the
filled triangles in Figure \ref{fig:brokersVliaisons}), we set all
ties to equal weight before normalizing for double-stochasticity,
resulting in stronger brokerage ties.

We find that the network with stronger brokers has a similar mixing
time to the full result of our algorithm, but it is much further from
optimum modularity than our result.  The network with weaker brokers
(equal to the weight of the liaison's ties) has slower mixing time
than our result, but it is closer to maximum modularity than the
network with strong brokers.

Overall, we observe that for a given rate of mixing, brokers produce a
less modular network than liaisons.  Alternatively, for a certain
amount of weight on inter-team ties, liaisons achieve a faster mixing
rate than brokers.

\section{Discussion}
\label{sec:discussion}

%
%
%

\subsection{Spectrally Specified Structures}

Through the spectral optimization technique we developed, we have
created networks with novel structural features.  In particular, what
we call ``liaisons'' --- individuals with strong connections in one
team and weaker connections to multiple (usually all) members of
another team --- maximize modularity, while maintaining a high degree
of inter-team connectivity.

\begin{figure}[tb]
	\centering
        \includegraphics[width=.25\textwidth]{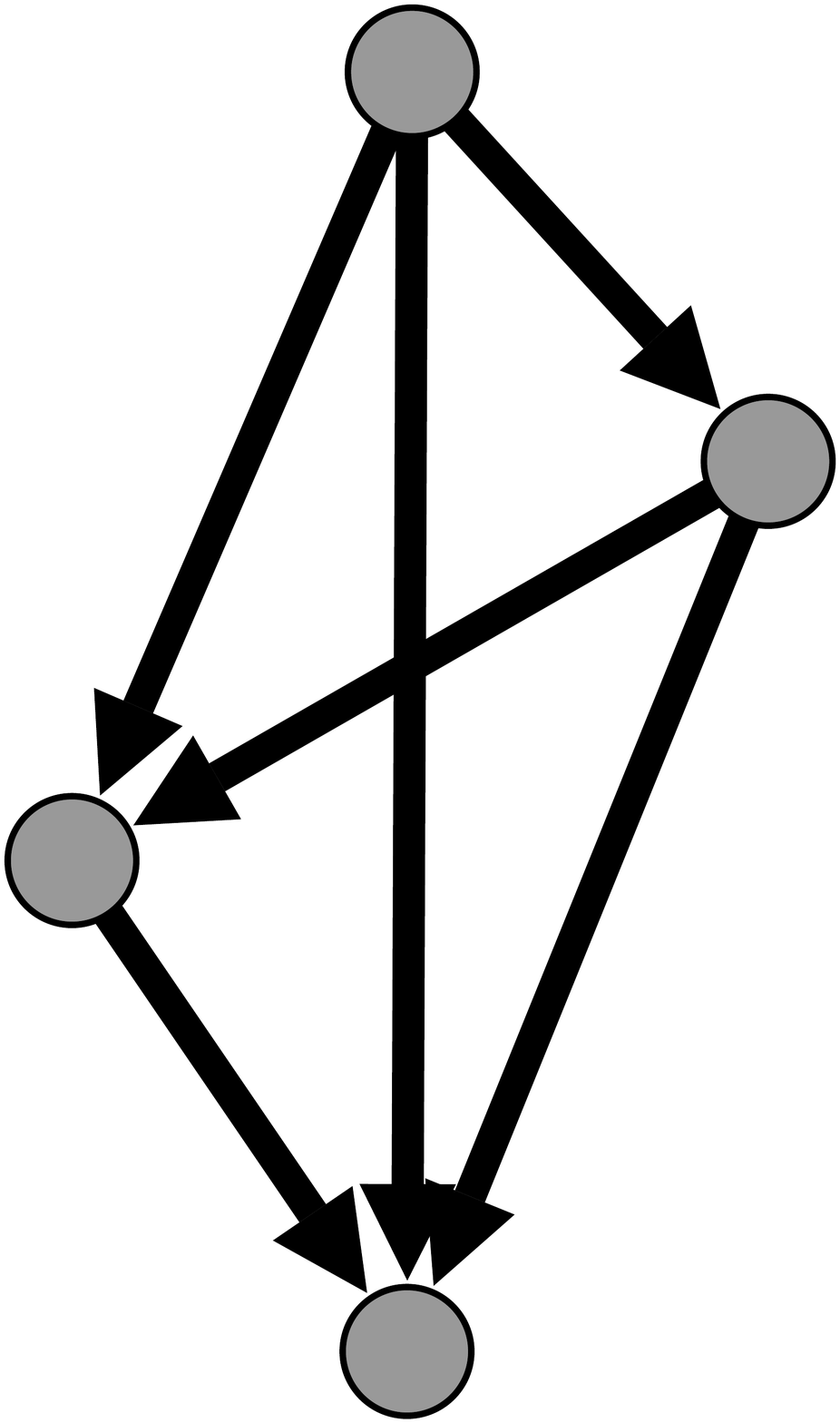}
	\caption{A simplified version of the right-hand graph in
          Figure~\ref{fig:clearExamples}, in which we have replaced
          cliques with nodes, and liaison membership in a secondary
          group with directed edges.  The network now forms a directed
          acyclic graph with the previously central module now at the
          root.}
	\label{fig:dag}
\end{figure}

We can elucidate the global structure of these networks, by
constructing simplified representational graphs.  To do this, we
collapse each connected clique to a single node and replace each
``fan'' of liaison edges by a single directed link.
Figure~\ref{fig:dag} shows the result of this process on the original
graph on the right side of Figure \ref{fig:clearExamples}, revealing
that the resulting structure is a directed acyclic graph, that is, a
strict hierarchy.
Accordingly, the networks admit interpretation such as the following.
The top of the hierarchy could be the leadership team.  At the other
end of the hierarchy, at the bottom of the figure, is a team that
receives representatives from all the other divisions, suggesting a
function depended on by all: perhaps an infrastructure or operations
team.

These communication structures are finely articulated, but this need
not present a barrier to implementation on a computer mediated
communication platform. One plausible implementation of tie strength
would be as a fraction of the problem solving time spent ``together,''
with the opportunity to exchange ideas or observe the progress of
others. On such a platform it would also be easy to tune the
importance of the weaker ties to increase either the speed of
information diffusion (with stronger ties), or the importance of
separate teams (with weaker ties), even over time to respond to the
collective progress within the network.

Thus, in platforms with a high degree of control over participant
ties, the network structures proposed here can be directly imposed on
them.
Alternatively, rather than implement the network directly, we could
follow a path inspired by \textit{mechanism design}
\citep[e.g.][]{myerson1988mechanism}, and create rules or technologies
by which self-interested individuals would choose to construct
spectrally optimal networks themselves.
We leave it to future work to determine which rules would achieve
these structures, and to analyze their dynamics, equilibria, and
properties.

\subsection{Spectral framework}

Beyond our argument for network design in general, we have identified
a framework for performing such design.  By building our framework
around the graph Laplacian spectrum, we have enabled the targeting of
many key network properties.  These include not only the modularity
and mixing time we have focused on here, but also degree distribution,
hierarchical structure, and many other graph properties.  Thus, by
defining objectives and constraints spectrally, and by providing a
non-convex optimization method for realizing graphs with such spectra,
a whole family of design problems become solvable.  Accordingly, there
are several further computational experiments that future research may
wish to investigate, and we highlight a few here.

In this paper, we assume that all individuals have equal
communications capacities, but in the real world, some people are more
capable than others.  To capture this variation, future experiments
could drop the requirement that networks be representable by
doubly-stochastic graphs, allowing some people to carry more
communicative weight.  Further, with suitable additions to the
optimization formulation, the matrix symmetry constraint could be
relaxed in order to investigate directed graphs in similar settings.

In this paper, we have looked at graphs with a single large step in
their spectra.  However, by including more than one step, graphs with
more than one level of hierarchical structure would be created.
Future work could investigate such graphs by including additional
spectral objective terms in the formulation presented here.

Of course, design work is only as good as the theory that supports it.
Design results should be tested experimentally using real human
problem solvers to assess their performance in crowdsourcing tasks and
knowledge management settings, relative to alternatives from the
literature.  Field work could also seek out empirical evidence that
deepens our understanding of how organizations that are similar to
networks proposed by our design framework operate in practice.
Relevant to the specific results we present here, for example, would
be study of organizations such as university administrations or the
United States government that utilize ``liaison officers'' to link
multiple divisions or teams together.  All of this further evidence
would be invaluable in informing future improvements to design
methodology.

\subsection{Conclusion}

Beyond describing networks, many practical settings call for designing
them.  For networks of humans, the benefits of network modularity have
been been well documented in research on networked problem solving,
above and beyond the well-explored benefits of short average path
lengths between all members of an organization.  However, prior work
on network design has not incorporated these insights. Our
contribution has been to fill this gap, drawing connections between
research on networked problem solving, spectral graph theory, and
non-convex optimization to both construct a design methodology and use
that methodology to generate novel structures for communication
networks.

\bibliographystyle{plainnat} 
\bibliography{inveigen} 

\end{document}